\documentclass[12pt,journal,letterpaper,onecolumn,oneside]{IEEEtran}

\IEEEoverridecommandlockouts  

\usepackage{times,latexsym,amsfonts,amsmath,amssymb,mathrsfs,bbm}
\usepackage{verbatim,cite}
\usepackage{epsfig,epsf}
\usepackage{graphics, graphicx}
\usepackage[usenames]{color}
\usepackage{tikz}
\usetikzlibrary{shadows,patterns,shapes,backgrounds,automata}
\usetikzlibrary{fit,decorations.pathmorphing,plotmarks} 
\tikzstyle{every picture}+=[remember picture] 

\usepackage[left=1in,right=1in,bottom=1in,top=1in]{geometry}

\def\indicator{{\mathbbm{1}}}

\newcommand{\prob}[1]{\mathsf{Pr} \left( #1\right)}
\newcommand{\EXP}[1]{\mathbb{E} \left(#1 \right)}
\newcommand{\exponent}[1]{\exp \left( #1 \right)}

\newcommand{\selfnote}[1]{}
\newcommand{\error}[1]{}
\newcommand{\ignore}[1]{}

\def\Interval{{\mathcal I}} 
\def\TargetLocationSet{{\mathcal T}} 

\def\TargetGrid{{\mathcal T}_{\mbox{\footnotesize
      g}}} 
\def\TargetRandom{{\mathcal T}_{\mbox{\footnotesize rnd}}}

\def\SensingRegion{{\mathcal R}} 
\def\ErrorEvent{{\mathcal E}}
\def\CouponErrorEvent{{E}}

\def\Adversaries{{\mathcal A}} 
\def\Sensors{{\mathcal S}} 

\newtheorem{theorem}{Theorem}
\newtheorem{lemma}{Lemma}
\newtheorem{remark}{Remark}

\title{On the Separability of Targets Using Binary Proximity
  Sensors$^*$\thanks{$^*$ This material is based upon work supported
    in part by the IRCC, IIT Bombay (Grant~No.~P09IRCC039), and the
    Bharti Centre for Communication at IIT Bombay.}}

\author{B.~Santhana Krishnan, Animesh Kumar, D.~Manjunath, and Bikash
  K.~Dey $^\dagger$\thanks{$^\dagger$ The authors are with the
    Electrical Engineering Department, IIT Bombay, Mumbai 400076,
    INDIA. Emails: {\tt\small
      \{skrishna,animesh,dmanju,bikash\}@ee.iitb.ac.in}} }

\begin{document}

\maketitle 

\thispagestyle{plain}
\pagestyle{plain}

\begin{abstract} 
  We consider the problem where a network of sensors has to detect the
  presence of targets at any of $n$ possible locations in a finite
  region. All such locations may not be occupied by a target. The data
  from sensors is fused to determine the set of locations that have
  targets. We term this the separability problem. In this paper, we
  address the separability of an asymptotically large number of static
  target locations by using binary proximity sensors. Two models for
  target locations are considered: (i) when target locations lie on a
  uniformly spaced grid; and, (ii) when target locations are
  i.i.d.~uniformly distributed in the area.  Sensor locations are
  i.i.d~uniformly distributed in the same finite region, independent
  of target locations. We derive conditions on the sensing radius and
  the number of sensors required to achieve
  separability. Order-optimal scaling laws, on the number of sensors
  as a function of the number of target locations, for two types of
  separability requirements are derived. The robustness or security
  aspects of the above problem is also addressed. It is shown that in
  the presence of adversarial sensors, which toggle their sensed
  reading and inject binary noise, the scaling laws for separability
  remain unaffected.
\end{abstract}

\section{Introduction}
\label{sec:intro}

Motivated by applications in cognitive radio, and in target sensing
situations like wildlife monitoring or land mine detection, we define
and develop the separability problem. An important requirement in
cognitive radio systems is the detection of white spaces---the regions
where the primary radio transmitters are not active. Consider the
following white space detection problem considered
in~\cite{Vaze:Whitespace:12}. In a region of interest, there are $n$
possible locations where these primary transmitters could be
present. It is reasonable to assume that each of these $n$ points may
contain at most one radio transmitter. To detect whitespace, i.e. the
area in where there is no radio reception, a set of radio receivers
are deployed randomly and each receiver can determine the existence of
a radio signal of strength above a specified threshold. The location
of the primary transmitters, and hence the available white space, is
to be determined using the binary output of the receivers.

As a second example, consider estimation of the population of rare
wildlife in a reserve forest. There are locations in these forests
that an animal is expected to visit e.g., watering hole or a salt
lick. If the animal is solitary, e.g., tigers or leopards, then at
most one of them will be present at any given time at any of these
locations. Sensors can be placed to sense the presence or absence of
an animal at these sites and the output from the sensors can be
used to estimate the population. Such a technique was employed to
estimate the tiger population in the Nagarahole reserve forest in
India~\cite{royleNKGA2009} where the forest was overlaid with an
approximate grid and sensors were suitably placed to sense the
presence of tigers in these sites. 

A third example is of land-mine detection. It is not unreasonable to
assume that, say, $n$ mines, have been randomly placed in an
area. Some of these are inert and others active. It is of interest to
detecting the location of the active mines using sensors that can
determine the presence of an active mine in their coverage range.

The preceding examples motivate the \emph{separability} problem, which
is defined next. A finite region of interest, say $\Interval,$ has $n$
points that are called \emph{target locations}. Each of these $n$
points contains at most one target. An ideal binary proximity sensor
of sensing radius $r(n)$ outputs a `1' if one or more targets are
present within its sensing radius $r(n)$ and outputs `0'
otherwise. $m(n)$ ideal binary proximity sensors are  randomly
deployed in $\Interval.$ The random location of the sensors models the
lack of precise control during sensor-deployment but the random
realization is assumed known. The objective is to find the
\emph{target configuration}---identify the set of target locations
that contain a target---using the outputs of these $m(n)$ sensors. We
determine order-optimal conditions on $r(n)$ and $m(n)$ to determine
the target configuration. This is a significant generalization of the
definition of separability described in~\cite{mudumbaiMI2008}. In this
paper we study several variations of the separability problem for the
two following models of target locations. 
\begin{enumerate}
\item \emph{Targets on grid}: where the target locations are on a
  uniform grid that is overlaid on $\Interval .$
\item \emph{Random target}s: where the target locations are
  i.i.d. realizations of a uniform random variable over $\Interval.$
\end{enumerate}
Clearly, the sensing radius of the binary proximity sensors determines
the quality of the separation that is achieved---a large sensing
radius lowers the resolution while a small sensing radius requires a
larger number of sensors. Thus the sensing radius is a design
parameter to be chosen suitably.

We are now ready to state the objective of this work---determine
$(r(n), m(n))$ the sensing radius of each sensor and the number of
sensors that are randomly deployed to achieve separability of the $n$
target locations. For each of the target location models, we seek to
find $r(n)$ and $m(n)$ for the following two performance criteria.
\begin{enumerate}
\item \textit{Full separability} where the configuration of all the
  $n$ target locations are to be identified correctly. Our results are
  asymptotic (in $n$) and have the form
  \begin{displaymath}
    \prob{ \mbox{all $2^n$ target configurations can be identified} }
    \to 1.
  \end{displaymath}
\item \textit{Partial separability} where the configuration of at
  least a fraction $\alpha, \ 0 < \alpha < 1,$ of the locations is to
  be determined correctly with probability at least $\beta, \ 0 <
  \beta < 1,$ i.e.
  \begin{displaymath}
    \prob{ \mbox{configuration at $\geq \ \alpha n$
        target locations are correctly identified} } \geq \beta.
  \end{displaymath}
\end{enumerate}

\subsection{Previous Work}
\label{sec:prior_art}
Localization of a source or a target is probably the closest class of
problems to separability. This is a very old problem and the
literature is replete with source and target localization using a
variety of measurement models. See~\cite{patwariAKHMCL2005} for an
excellent survey of localization problems in sensor networks. While a
large part of the localization literature considers measurement models
like range, angle-of-arrival, etc, binary proximity sensors have also
been used in several localization problems e.g.,
\cite{Simic01,Karnik04}. More recently, binary proximity sensors have
also been used in target tracking, e.g.,
\cite{singhMKSCT2007,shrivastavaMMST2009}.  Another problem closely
related to separability is the counting problem---count the number of
targets in a finite sensing area~\cite{gandhiKST2008}. In
\cite{gandhiKST2008}, the counting problem has been studied with
sensors that can output the number of distinct targets they can sense,
i.e., the output is not binary. We will see below that separability is
distinct from both of these.

The notion of separability was introduced in \cite{mudumbaiMI2008}
where the following problem was studied. A single target is located at
one of two possible locations, say $t_1$ and $t_2.$ Binary proximity
sensors, possibly non ideal, are deployed in $\Re^2$ according to a
spatially homogeneous Poisson process of density $\lambda.$ The
separability problem, identifying which of $t_1$ and $t_2$ contains
the target, was formulated as a binary hypothesis testing problem and
fundamental bounds on the decoding error was obtained using
information theoretic techniques. They also consider the case where
the sensor output is from an alphabet $\mathcal{Y}.$ The difference
between separability and localization is now apparent---separability
is a disambiguation problem while localization is an estimation
problem.  In~\cite{mudumbaiMI2008} it is assumed that the target is
present in exactly one of two possible locations; we generalize and
consider the case where upto one target can be present at each of $n$
locations. Thus our disambiguation is between the $2^n$ possibilities,
akin to decoding.

Much of our techniques and results will be closely related to results
in coverage problems. It may be noted that building on coverage
problems outlined in \cite{Hall88}, there has been a significant
amount of work on coverage in sensor networks,
e.g. see~\cite{Kumar04,Liu04}. The primary interest in this line of
research is to use random shapes (sensor coverage areas) and cover any
subset of $\Re^d$ or a measurable fraction of the subset. Infer that
the separability problem reduces to the coverage of a countable number
of points with extra restrictions, we will compare our results to
analogous results from coverage analysis.

\subsection{Organization of the Paper and Summary of Results}

The rest of the paper is organized as follows. The system model and
relevant mathematical results are described in
Section~\ref{sec:keyideas}. The main results, i.e., the scaling laws
for critical $r(n)$ and the corresponding $m(n)$ for the two target
models (Theorem~\ref{thm:separability_deterministic} and
\ref{thm:separability_random}) are described in
Section~\ref{sec:scalinglaws}. In Subsection~\ref{sec:fixed_targets}
we consider the targets-on-grid model and randomly realized target
locations are described in Subsection~\ref{sec:uniformtargets}. For
pedagogical convenience, Section~\ref{sec:scalinglaws} will deal with
separability on $\Interval=[0,1]$ and the two dimensional extension is
described in Theorem~\ref{thm:2d} in Section~\ref{sec:2d}.

For secure settings, it is also desirable to have some form of
robustness against adversarial sensors; an adversarial sensor can
mislead the decision process by injecting binary noise, that toggles
its actual reading. This form of adversarial sensing is discussed in
Section~\ref{sec:adversaries} where we assume that there is a known
upper bound on the fraction of sensors that are adversarial. We will
argue in Theorem~\ref{thm:adversary:grid} that majority logic can be
used and the order of $r(n)$ and $m(n)$ does not change. Finally,
conclusions are presented in Section~\ref{sec:conclusions}.

\section{System Model and Mathematical Preliminaries}
\label{sec:keyideas}

In this section, we describe the system model and relevant
notation. This is followed by some known mathematical results which
will be used in the subsequent sections.

\subsection{System model}
\label{sec:system_model}

The sensor field is a finite interval $\Interval ;$ without loss of
generality we assume that, $\Interval = [0, 1]$.  $\TargetLocationSet$
is the set of $n$ $(n < \infty)$ distinct points in $\Interval$ that
are the target locations. Two models for $\TargetLocationSet$ will be
used in this work. In the targets-on-grid model, the target
locations ($\TargetGrid$) are on a finite grid, i.e.,
\begin{displaymath}
  \TargetGrid := \left\{ \frac{1}{2n}, \frac{3}{2n}, \ldots, \frac{(2n 
      - 1)}{2n} \right\}. 
\end{displaymath}
In the random-targets model, the target locations ($\TargetRandom$) are
a realization of $n$ i.i.d.~random variables uniformly distributed in
$\Interval.$ They will be represented using the ordered target
locations as below.
\begin{displaymath}
  \TargetRandom := \left\{ T_{(1)}, T_{(2)}, \ldots, T_{(n)} \right\}
\end{displaymath}
Here $\{T_{(i)}, n \in \mathbb{N}\}$ is the $i$-th order statistic of
$n$ i.i.d.~$\mbox{Uniform}[0,1]$ random variables. We reiterate that
all target locations in $\TargetLocationSet$ need not be occupied by
targets.

Recall that, an ideal binary proximity sensor at location $x$ with
sensing radius $r(n)$ outputs a $1$ if and only if there exists at
least one target in $(x-r(n), x+ r(n)).$ The locations of the set of
$m(n)$ sensors is denoted by $\{X_1, X_2, \ldots, X_{m(n)}\},$ where
$X_i$ are i.i.d.~uniformly distributed in $\Interval.$ Throughout the
paper, we assume all sensors to be ideal binary proximity sensors. To
detect the possible presence of targets in $\TargetLocationSet,$
$m(n)$ sensors are randomly deployed in $\Interval.$ Each sensor has a
sensing radius of $r(n)$ i.e., for a sensor at location $x$, the
sensing region is
\begin{eqnarray*}
  \SensingRegion(x, r(n)) = \{ y : y \in \Interval \mbox{ and } |y -
  x| < r(n)\}.\label{eq:sensingregion}
\end{eqnarray*}
The sensing radius $r(n)$ will be treated as a \textit{design
  parameter}. 

The data recording model of the sensors is as follows.  A sensor at
$x$ outputs a logical $1$ if it detects at least one target in
$\SensingRegion(x, r(n))$. We will see from the following argument
that target location $T_i$ is unambiguously identifiable by a sensor
if and only if the sensor detects $T_i$ and no other $T_j, j \neq i.$
Since we assume that targets can be present only at the target
locations in $\TargetLocationSet,$ the following cases prove the above
claim.
\begin{enumerate}
\item For a sensor at $x,$ if $ T_i \notin \SensingRegion(x, r(n)) \
  \forall \ i \in \left\{1,\ldots,n\right\},$ then it outputs a
  logical `0' irrespective of the target configuration and the sensor
  observation is not useful. The sensor at $x_a$ in
  Fig.~\ref{fig:TargetsAndSensingII} illustrates this condition.
\item \label{unique:coverage} For a sensor at $x,$ and some $i$ and
  $j,$ $1 \leq i < j \leq n,$ let $T_i \in \SensingRegion(x, r(n))$
  and $T_j \in \SensingRegion(x, r(n)).$ If at least one of $T_i$ or
  $T_j$ has a target then the sensor at $x$ will output a
  `1'. However, this sensor's observation cannot be used to
  distinguish any configuration of $T_i$ and $T_j$ with at least one
  target. The sensor at $x_b$ and target locations $T_a$ and $T_b$ in
  Fig.~\ref{fig:TargetsAndSensingII} illustrate this condition.
\item Let three consecutive target locations $\{T_{i-1}, T_i, T_{i+1}
  \}$ be such that $\vert T_i - T_{i-1} \vert < r(n)$ and $\vert T_i -
  T_{i+1} \vert < r(n),$ then all sensors that cover $T_i$ also cover
  either $T_{i-1}$ or $T_{i+1}.$ If there is a target at both
  $T_{i-1}$ and $T_{i+1},$ then the presence or absence of a target at
  $T_i$ cannot be distinguished by any set of sensors.
\end{enumerate}
\begin{figure}
  \begin{center}
    \scalebox{1.0}{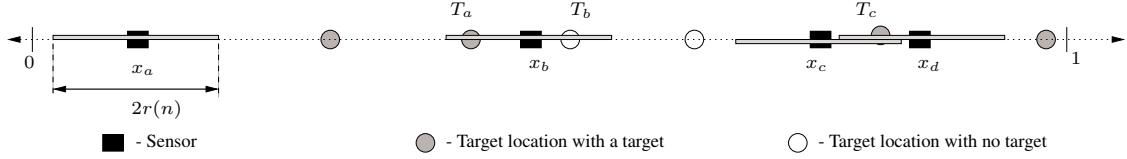} 
    \caption{Illustrating conditions on identifiability of a target
      locations. Target location $T_c$ is identifiable. Target
      locations $T_a$ and $T_b$ are not identifiable.}
    \label{fig:TargetsAndSensingII}
  \end{center}
\end{figure}
Thus, a sensor at $x$ can be used to determine the target
configuration at $T_i$ if and only if $T_i \in \SensingRegion(x,
r(n))$ and $T_j \notin \SensingRegion(x, r(n)) \ \forall \ j \neq i.$
This leads us to the following definition. We say that target location
$T_i$ is \textit{identifiable} if there is at least one sensor at $x
\in \Interval$ such that $T_i \in \SensingRegion(x, r(n))$ and $T_j
\notin \SensingRegion(x, r(n)), \ \forall \ j \neq i.$ The target
location at $T_c$ in Fig.~\ref{fig:TargetsAndSensingII}, is covered by
sensors at $x_c$ and $x_d,$ and is identifiable. Thus full
separability is equivalent to having $n$ identifiable targets.

In this paper we seek two types of separability results. In
\textit{full separability}, the objective is to determine the
asymptotic $r(n)$ and $m(n)$ for which every possible target
configuration is separated with high probability. In other words, find
$r(n)$ and $m(n)$ that will, with high probability, identify every
target location. The second set of results determine the $r(n)$ and
$m(n)$ to achieve \textit{partial separability}, i.e., we determine
these quantities for which at least $\alpha n,$ $0 < \alpha < 1,$
target locations are identifiable with a probability at least $\beta,$
$0 < \beta < 1$.

To make this paper self-contained, we next present some mathematical
results, some of which are known in the literature, that will be used
in our analysis.

\subsection{Mathematical preliminaries}
\label{sec:math}
First, a note on symbols. The set of reals and naturals are denoted by
$\Re$ and $\mathbb{N}$ respectively. We have already used the symbol
$\mathsf{Pr}$ for probability of an event; it is assumed that there is
a common $(\Omega, {\mathcal F}, \mathsf{Pr})$ structure for defining
all the events in this work. \ignore{rephrase this suitably. what is
  the space or do we need a space, cant we create a product space from
  all events? -- expect some trouble if ever it is.. }

The order notation is well known but we recapitulate them here for
completeness. For positive sequences $f(n)$ and $h(n)$ we say that
$f(n) = \Theta(h(n))$ if there are non-zero positive constants $0 <
a_1 < a_2$ and a corresponding $N \in \mathbb{N}$ such that for all $n
\geq N,$ $a_1 h(n) \leq f(n) \leq a_2 h(n)$.  Similarly, we say that
$f(n) = \omega(h(n))$ if $\lim_{n \rightarrow \infty} f(n)/h(n) =
\infty$.

\ignore{The first result below establishes the closeness of
  $(1-\theta)^m$ to an exponential function of $\theta$ and $m$. The
  second result is on the order statistics of uniform random
  variables. Following this, we discuss some variants of the Markov
  inequality. Finally, the coupon collector
  problem~\cite{motwaniRR1995} is discussed.}

The following lemma bounds the asymptotic behavior of $(1-\theta)^m.$
\begin{lemma}
  For constant $\theta,$ $0 < \theta < 1$ and any positive integer
  $m,$
  \begin{eqnarray}
    \exponent{ - \frac{m \theta}{1 - \theta} } \ < \ (1 - \theta)^m \
    <  \ \exponent{ - m \theta }. \label{eq:braveinequality}
  \end{eqnarray}
  This implies that $(1 - \theta)^m \rightarrow 0$ if and only if
  $\exponent{-m\theta} \rightarrow 0.$
\end{lemma}
\begin{proof} For $0 < \theta < 1,$ $(1 - \theta) < \exp(-\theta),$
  and thus the upper bound follows. The lower bound is obtained from
  $\exp( x ) > 1 + x,$ for $x > 0,$ using $x = \frac{\theta}{1-
    \theta}.$
\end{proof}

The following results from order statistics are adapted from
\cite[pg.~134]{davidno2003}.  Let $\{U_i, 1 \leq i \leq n\}$ be
i.i.d.~$\mbox{Uniform}[0,1]$ random variables and $U_{(i)}, 1 \leq i
\leq n$ be their order statistics, i.e., $U_{(1)} \leq U_{(2)} \leq
\cdots \leq U_{(n)}.$ Let $U_{(0)} := 0.$ Define the spacing variables
as follows: $V_1 = U_{(1)}$, $V_i = U_{(i)} - U_{(i-1)}, \ 2 \leq i
\leq n$, and $V_{n+1} \ = \ 1 - U_{(n)} \ = \ 1 - \sum_{i = 1}^n V_i.$
The joint probability density function of $\{V_i, 1 \leq i \leq n\},$
\begin{equation*}
  f_{V_1, \ \ldots \  V_{n}} (v_1, \ldots, v_{n}) = 
  \begin{cases}
    n! & \mbox{for $v_i \geq 0$ and $\sum_{i = 1}^n v_i \leq 1,$} \\
    0 & \mbox{otherwise.}
  \end{cases}
\end{equation*}
Therefore,
\begin{equation}
  \prob{V_1 > v_1, \ldots, V_{n} > v_{n} } = 
  \begin{cases}
    (1 - v_1 - \ldots - v_{n})^{n}, & \mbox{if $\sum_{i = 1}^n v_i
      \leq 1,$} \\
    0 & \mbox{otherwise.}
  \end{cases}
  \label{eq:spacingdistribution}
\end{equation}
Since the probability density function $f_{V_1, \ \ldots \ V_{n}}
(v_1, \ldots, v_{n})$ is symmetric, the distribution of any $k$
spacings, $1 \leq k \leq n,$ has the same distribution as that of the
first $k$ spacings, i.e., of $V_1,\ldots, V_k.$ This is obtained by
setting the $v_i = 0$ for the other $(n-k)$ spacings
in~\eqref{eq:spacingdistribution}. Thus, for any $k < n$ and $1 \leq
n_1 < n_2 < \ldots n_k \leq n$,
\begin{equation}
  \prob{V_{n_1} > v_1, \ldots, V_{n_k} > v_{k}} = 
 \begin{cases}
   (1 - v_1 - \ldots - v_{k})^{n} & \mbox{if $v_i > 0$ and 
     $\sum_{i =1}^k v_i \leq 1,$} \\
   0 & \mbox{otherwise.}
 \end{cases}
 \label{eq:kmarginals}
\end{equation}

Next, we derive a version of the Markov inequality on the sum of
Bernoulli random variables. Let $\{B_n \},$ $n > 0,$ be a sequence of
i.i.d. Bernoulli random variables with parameter $p$ and let $S_n :=
\sum_{i= 1}^n B_i.$ From Markov inequality on $n-S_n,$ we have
\begin{eqnarray}
  \prob{(n - S_n) \geq (1-\alpha) n} & \leq & 
  \frac{(1-p)n}{(1-\alpha)n} \nonumber \\
  \prob{S_n \leq n \alpha} & \leq &  \frac{1-p}{1-\alpha} 
  \nonumber \\
  \prob{S_n > n \alpha} & = & 1 - \prob{S_n \leq \alpha n} \ \geq \
  \frac{p-\alpha}{1-\alpha} .
  \label{eq:reversemarkov}   
\end{eqnarray}
Observe that this bound is uniform in $n.$

We now summarize some results for the coupon collector
problem~\cite{motwaniRR1995}. Recall that in the coupon collector
problem, there are $n$ distinct coupons in a bag and coupons are
sampled with replacement.  The quantity of interest is the minimum
number of samples so that each coupon is sampled at least once. Let
$\CouponErrorEvent_{i}$ indicate that coupon $i$ has not been sampled
in $m$ draws. The following equality relations asymptotically
hold~\cite{motwaniRR1995}.
\begin{enumerate}
\item For any constant $c > 0,$ if $m = n (\log n - c)$ then $\lim_{n
    \rightarrow \infty} \prob{\sum_{i=1}^n \CouponErrorEvent_i \geq 1}
  = 1 - \exp(-\exp(c)).$ If instead of $c,$ we use any $c_n \to
  \infty,$ then $\lim_{n \to \infty} \exp(-\exp(c_n))= 0$ and
  $\prob{\sum_{i=1}^n \CouponErrorEvent_i \geq 1 } \to 1.$
\item If $m = n (\log n + c),$ then $\lim_{n \rightarrow \infty}
  \prob{\sum_{i=1}^n \CouponErrorEvent_i \geq 1} = 1 -
  \exp(-\exp(-c)).$ The following hold.
  \begin{enumerate}
  \item For any real positive constant $c,$ $\lim_{n \rightarrow
      \infty} \prob{\sum_{i=1}^n \CouponErrorEvent_i \geq 1} < 1.$
  \item If instead of $c,$ we use any $c_n \to \infty,$ then $\lim_{n
      \to \infty} \exp(-\exp(-c_n))= \ignore{\exp(-0) =}1.$ Thus
    $\prob{\sum_{i=1}^n \CouponErrorEvent_i \geq 1} \to 0.$
  \end{enumerate}
\end{enumerate}
Thus, if $m$ is the number of samples needed to sample all of the $n$
coupons, then $m \geq n (\log n + c_n)$ for any $c_n \to \infty.$
Observe that to ensure that every coupon has been drawn $m/n$ is
logarithmic.

The main results are presented in the next section.

\section{Scaling laws for separability}
\label{sec:scalinglaws}
We first consider the targets-on-grid model and then consider the
random targets model.
\subsection{Separability of target locations on a grid}
\label{sec:fixed_targets}
For notational convenience, the targets in $\TargetGrid$ will be
numbered $1, 2, \ldots, n$ from the left. Recall that in
$\TargetGrid,$ the $i$-th target location $T_i = (2i -1)/(2n), 1 \leq
i \leq n.$ Theorem~\ref{thm:separability_deterministic} presents the
results of this subsection.
\begin{theorem}[Separability of targets-on-grid]
  \label{thm:separability_deterministic}
  For the sensing region $\Interval = [0,1],$ and target locations
  $\TargetGrid,$ when $m(n)$ sensors are deployed uniformly in
  $\Interval,$
  \begin{enumerate}
  \item $0 < r(n) < (1/n)$ is necessary for
    separability. \label{fixed:sep:rn}
  \item Let $r(n) = a/2n,$ or $r(n) =(2-a)/2n,$ for $0 < a \leq 1,$
    then the following are true.~\label{fixed:full:mn}
    \begin{enumerate}
    \item If $m(n) \geq (n/a) \ (\log (n/a) + c_n),$ for any $c_n \to
      \infty,$ then \newline $\prob{\mbox{all target configurations
          are separable} } \to 1.$ \ignore{\item If $m(n) \leq (n/a) \
        (\log n + c),$ for any real $c > 0,$ then $\prob{\mbox{all
            target configurations are separable}} < 1.$}
    \item If $m(n) \leq (n/a) \ (\log (n/a) - c_n),$ for any $c_n \to
      \infty,$ then \newline $\prob{\mbox{all target configurations
          are separable} } \to 0.$ \label{fixed:full:mn:nec}
    \end{enumerate}
  \item Let $r(n) = a/2n,$ or $r(n) = (2-a)/2n,$ for $0 < a \leq 1.$
    Given $0 < \alpha < 1,$ and $0 < \beta < 1,$ the following are
    true.\label{fixed:partialmn}
    \begin{enumerate}
    \item If $m(n) \geq (n/a) \ \log \left( \frac{1}{(1 - \alpha) (1 -
          \beta)} \right) $ then $\prob{\mbox{at least } \alpha n
        \mbox{ targets are separable}} >
      \beta.$ \label{fixedpartialsuff}
    \item If $m(n) < (n/a - 1) \ \log \left( \frac{1}{(1 - \alpha
          \beta)} \right) $ then $\prob{\mbox{at least } \alpha n
        \mbox{ targets are separable}} <
      \beta.$\label{fixed:partial:nec}
    \end{enumerate}
  \end{enumerate}
\end{theorem}
\begin{proof}
  We first prove statement~\ref{full:sep:rn}.
  \begin{itemize}
  \item If $r(n) > 1/n$, then all sensors have at least two target
    locations in their sensing region. From our discussion in
    Section~\ref{sec:system_model}\ignore{(pg.~\pageref{unique:coverage})},
    it follows that no target is identifiable.
  \item If $r(n) = 1/n,$ then only sensors placed at $T_i$ sense
    exactly one target while any sensor at other locations senses two
    target locations. In a random sensor deployment, having sensors at
    the target locations $\TargetGrid$ has zero probability, thus
    targets are not separable with probability $1.$ \ignore{Clearly,
      having sensors with $r(n) = 1/n$ that can identify any target
      $T_i$ is a zero probability event.}
  \end{itemize}
  This proves statement~\ref{full:sep:rn} that $0 < r(n) < 1/n$ is
  necessary for separability. Next we prove
  statement~\ref{full:sep:mn}.

  First let $r(n) = a/2n,$ with $0 < a \leq 1.$ From the uniform
  distribution of the sensors, a sensor covers target location $i$
  with probability $a/n.$ For $r(n) \leq 1/2n,$ if a target location
  is covered, then it is identifiable. Thus with the $n$ target
  locations as coupons and the $m(n)$ sensors as draws, this is
  analogous to the coupon collector problem. For full separability we
  need all the target locations to be covered by at least one sensor;
  hence statement~\ref{full:sep:mn} follows.

  Next consider $r(n) = (2-a)/2n$ with $0 < a \leq 1.$ For each $T_i,$
  any sensor in the interval $\Interval_i := (T_i - (a/2n), T_i +
  (a/2n))$ covers only $T_i$ while a sensor elsewhere that covers
  $T_i$ will also cover $T_{i-1}$ or $T_{i+1}.$ From our discussion in
  Section~\ref{sec:system_model}\ignore{(pg.~\pageref{unique:coverage})},
  $T_i$ is identifiable if and only there is at least one sensor in
  $\Interval_i.$ The probability that a uniformly deployed sensor node
  falls in $\Interval_i$ is $a/n,$ which is the same as that for $r(n)
  = a/2n$ for $0 < a \leq 1.$ The rest of the proof for
  $r(n)=(2-a)/2n$ follows analogous to the case of $r(n) = a/2n.$ This
  completes the proof of statement~\ref{full:sep:mn}.

  Before we prove statement~\ref{fixed:partialmn}, we first derive
  upper and lower bounds on the probability of having at least $\alpha
  n$ identifiable target locations. First, let $r(n) = a/2n, 0 < a
  \leq 1.$ Recall that for this $r(n),$ a target is identifiable if
  and only if it is covered. Let $\ErrorEvent_i$ be the indicator of
  the event that $T_i$ is not covered. For partial separability, we
  require $\prob{ \sum_{i = 1}^n (1-\ErrorEvent_i) \ \geq \ \alpha n }
  \geq \beta.$ Using \eqref{eq:reversemarkov} and the Markov
  inequality on $(1-\ErrorEvent_i),$ we have the following lower and
  upper bounds respectively. For notational convenience, we will use
  $m$ instead of $m(n)$ for the rest of this proof.
  \begin{equation}
    \frac{1 - \left(1 - (a/n)\right)^m - \alpha}{1 - \alpha} \leq
    \prob{ \sum_{i = 1}^n (1-\ErrorEvent_i) \geq \alpha n } \leq \frac{1 -
      \left(1 - (a/n)\right)^m}{\alpha}. \label{eq:temp1}
  \end{equation}
  Now, we prove statement~\ref{fixedpartialsuff}. Let $m(n)$ be chosen
  such that $m \geq \left(\frac{n}{a}\right) \log \left( \frac{1}{(1 -
      \alpha)(1- \beta)} \right).$ By appropriate manipulations
  and~\eqref{eq:braveinequality}, the lower bound in~\eqref{eq:temp1}
  is $> \beta$ as shown below.
  \begin{displaymath}
    \left(1 - (a/n)\right)^m \ < \ \exponent{-\frac{am}{n}} \leq (1 -
    \alpha)(1- \beta) \ \iff \ 1 - \frac{\left(1 - (a/n)\right)^m}{1 -
      \alpha} > \beta.
  \end{displaymath}
  This completes the proof of statement~\ref{fixedpartialsuff}. Next,
  we prove statement~\ref{fixed:partial:nec}
  using~\eqref{eq:temp1}. Let $m(n)$ be chosen such that $m <
  \left(\frac{n}{a} - 1 \right) \log \left( \frac{1}{1-\alpha
      \beta}\right).$ By algebraic manipulations and
  using~\eqref{eq:braveinequality}, we see that the upper bound
  from~\eqref{eq:temp1} is $< \beta$ as shown below.
  \begin{displaymath}
    \left(1 - (a/n)\right)^m \ > \ \exponent{-\frac{m}{\frac{n}{a}-1}} > 1
    -\alpha \beta \ \iff \ \frac{1 - \left(1 - (a/n)\right)^m}{\alpha} <
    \beta.
  \end{displaymath}
  This completes the proof of statement~\ref{fixed:partialmn} for
  $r(n) = a/2n,$ with $0 < a \leq 1.$ If $r(n) = (2-a)/2n,$ for $0 < a
  \leq 1,$ then by an argument identical to the proof of
  statement~\ref{fixed:full:mn}, the conditions on $m(n)$ are
  identical to that with $r(n) = a/2n.$
\end{proof}
\begin{remark}
  If $r(n) = 1/2n$ then full coverage of $\Interval,$ as defined
  in~\cite{Hall88}, is a sufficient condition for full
  separability. Note that in the coverage analysis in~\cite{Hall88},
  sensors are distributed according to a homogeneous spatial Poisson
  process of intensity $\lambda(n).$ From~\cite[(2.24) and Thm
  3.11]{Hall88}, $\lambda(n) = c n \log n$ with $c > 1$ is necessary
  and sufficient for full coverage of $\Interval.$ Observe that the
  constant factor multiplying the $n \log n$ term for full coverage is
  $c > 1$ while for full separability it is $c = 1.$
\end{remark}
\begin{remark}
  If sensors have sensing radius $r(n) = 1/2n,$ then using the Markov
  inequality and~\cite[(3.11)]{Hall88}, we can show that to cover at
  least $\alpha, \ 0 < \alpha < 1,$ length of $\Interval$ with
  probability at least $\beta,$ the necessary and sufficient
  conditions on $\lambda(n)$ are identical to those of $m(n)$ obtained
  in statement~\ref{fixed:partialmn} of
  Theorem~\ref{thm:separability_deterministic}. Thus partial coverage
  and partial separability have identical requirements on the sensor
  density.
\end{remark}
\begin{remark}
  The sensing radius $r(n) = 1/(n+1),$ does satisfy
  statement~\ref{fixed:sep:rn} of
  Theorem~\ref{thm:separability_deterministic}, but in that case, the
  $m(n)$ required will be such that $m(n) \in \Theta \left(n^2 \log n
  \right)$ for full separability in fixed grid model.
\end{remark}

\subsection{Separability of uniformly distributed target-locations}
\label{sec:uniformtargets}

In this subsection, the target locations $\TargetRandom$ are
distributed uniformly in $\Interval.$ \ignore{Consider a realization
  of the target locations and number them from the left.} For
notational convenience in this subsection we use $T_i$ instead of
$T_{(i)}.$ In a realization, target location $T_i$ may not be
separable due to either of the following reasons.
\begin{enumerate}
\item $(T_i - T_{i-1})$ and $(T_{i+1} - T_i)$ are both less than
  $r(n),$ and no sensor can identify $T_i.$
\item There are no sensors uniquely covering $T_i.$
\end{enumerate}
As in the previous subsection, we seek $r(n)$ and $m(n)$ to achieve
full and partial separability and will account for both these failure
conditions. Theorem~\ref{thm:separability_random} is the main result
of this subsection.
\begin{theorem}[Separability of random target locations]
  \label{thm:separability_random}
  For the sensing region $\Interval = [0,1]$ and target locations
  $\TargetRandom,$ $m(n)$ sensors are deployed uniformly i.i.d.~in
  $\Interval.$
  \begin{enumerate}
  \item \ignore{\selfnote{Here Prof. BKD feels that we could write
        this result as a pro and a converse result, in two parts. One
        to say that if $r(n) = 1/c n^2,$ then $\prob{\min W_i \geq 2
          r(n)} <1,$ and the other when $r(n) = 1/ c_n n^2$ for some
        $c_n \to \infty.$}} For full separability of $n$ target
    locations, it is necessary that $r(n) = 1/(c_n n^2)$ for some $c_n
    \to \infty.$~\label{full:sep:rn}
  \item Let $r(n) = 1/(c_n n^2),$ for some $c_n \to \infty.$ Then the
    following are true.~\label{full:sep:mn}
    \begin{enumerate}
    \item For any $f_n \to \infty$ if $m(n) \geq \frac{1}{2r(n)}
      \left( \log \left( \frac{1}{2r(n)}\right) + f_n \right) = \left(
        \frac{n^2 c_n}{2} \right) (2 \log n + \log \left(c_n/2 \right)
      + f_n ),$ then $\prob{\mbox{all } n \mbox{ target locations are
          separable}} \to 1.$
    \item For any $f_n \to \infty$ if $m(n) \leq \frac{1}{2r(n)}
      \left( \log \left( \frac{1}{2r(n)}\right) - f_n \right) = \left(
        \frac{n^2 c_n}{2} \right) (2 \log n + \log \left(c_n/2 \right)
      - f_n),$ then $\prob{\mbox{all } n \mbox{ target locations are
          separable}} \to 0.$
    \end{enumerate}
  \item \ignore{\selfnote{Prof. BKD feels that this statement is not
        in sync with thm 1 and the other statements. This is a result
        similar to $r_n$ for full separability, but it is a type of
        lemma used in the proof of next step. So should this be
        included here this way or can we write this as a part of the
        next statement?} } For any $0 < \alpha_1, \beta < 1,$ let $c_1
    := \log \left( 1/ \left( 1 - \left(1 - \alpha_1\right)\left(1 -
          \beta\right) \right) \right).$~\label{partial:sep:rn}
    \begin{enumerate}
    \item If $r(n) \leq \frac{1}{2 \left(\frac{n}{c_1} + 1 \right)},$
      then \newline $\prob{\mbox{at least } \alpha_1 n \mbox{ targets
          are more than $r(n)$ away from adjacent neighbors}} \geq
      \beta.$ ~\label{partial:sep:rn:suff}
    \item If $r(n) > \frac{ \log \left(\frac{1}{\alpha_1 \beta}
        \right)}{2n},$ then \newline $\prob{\mbox{at least } \alpha_1
        n \mbox{ targets are more than $r(n)$ away from adjacent
          neighbors}} < \beta.$~\label{partial:sep:rn:nec}
    \end{enumerate}
  \item For a given $0 < \alpha, \beta < 1,$ choose an $\alpha_1$ such
    that $\alpha < \alpha_1 < 1.$ Let $c_2 := \log $ $\left( 1/
      \left( 1 - \left(1 - \alpha\right)\left(1 - \beta\right) \right)
    \right),$ $c_3 := \log \left( 1/(\alpha \beta) \right)$ and $c_1$
    is as defined in statement~\ref{partial:sep:rn} above. Let
    $\theta_1 \left(c_1/ \left(2 n \right) \right) \leq r(n) \leq
    \theta_2 \left(c_1/ \left(2 n \right) \right),$ for any $\theta_1,
    \theta_2$ such that $0 < \theta_1 \leq \theta_2 < 1/ \left(1 +
      \left( c_1/ n\right) \right).$ Choose a finite positive constant
    $a$ such that $a > \max \left\{1, c_2/ (2 \theta_1 c_1)
    \right\}.$~\label{partial:sep:mn}
    \ignore{Prof. BKD feels that this way of writing $r(n)$ gives it
      much more flexibility in terms of the function
      definition. Earlier I wrote $r(n) = c_1/ 4n,$ in that type of
      statement, it just implies it has to be a function that is going
      down by a constant times $1/n$ in this form, it says, it can be
      any arbitrary way as long as it is between two functions that go
      down in a inverse linear fashion. }
    \begin{enumerate}
    \item If $m(n) \geq \left( \frac{n}{\theta_1 (a- 1)c_1} \right)
      \log\left(1 + \frac{1}{c_2 - 2 a \theta_2 c_1} \right),$ then
      \newline $\prob{\geq \alpha n \mbox{ target locations are
          separable}} \geq \beta.$~\label{partial:sep:mn:suff}
    \item If $m(n) < \left( \frac{n}{\theta_2 (a-1)c_1} - 1 \right)
      \log \left( \frac{1}{c_3 - a \theta_1 c_1} \right),$ then
      \newline $\prob{\geq \alpha n \mbox{ target locations are
          separable}} < \beta.$~\label{partial:sep:mn:nec}
    \end{enumerate}
  \end{enumerate}
\end{theorem}
Before the proof of Theorem~\ref{thm:separability_random}, we first
characterize the minimum separation between adjacent target
locations. Recall the definition of spacings, $V_i = T_{i} - T_{i-1},
\ 1 \leq i \leq n$ with $T_0 := 0.$ Lemma~\ref{lemma:rnd:min:dist}
characterizes $\min_{2 \leq i \leq n} V_i$ and is necessary to prove
statement~\ref{full:sep:rn} of Theorem~\ref{thm:separability_random}.
\begin{lemma}
  \label{lemma:rnd:min:dist}
  Let $c_n$ be any sequence such that $c_n \to \infty.$ For any
  sequence $d_n,$ such that $0 < d_n < 1/n,$ $\prob{\min_{2 \leq i
      \leq n} V_i \geq d_n} \to 1$ if and only if $d_n = \frac{1}{c_n
    n^2}.$
\end{lemma}
\begin{proof}
  From~\eqref{eq:spacingdistribution} we have:
  \begin{equation*}
    \prob{\min_{2 \leq i \leq n} V_i > d_n } = \prob{V_2 > d_n, V_3 >
      d_n, \ldots, V_n \geq d_n} = (1 - (n-1) d_n)^{n}.
  \end{equation*}
  Upper and lower bounds on the preceding probability
  using~\eqref{eq:braveinequality} are given below.
  \begin{equation*}
    \exponent{ - \frac{n (n-1) d_n}{1 - (n-1)d_n} } \leq
    \prob{\min_{2 \leq i \leq n} V_i > d_n} \leq \exponent{ - n (n-1) d_n}.
  \end{equation*}
  The `if' part of Lemma~\ref{lemma:rnd:min:dist} is proved as
  follows, let $d_n = 1/ (c_n n^2),$ for some $c_n \to \infty.$ Then
  $\exponent{- \frac{n(n-1) d_n}{1 - (n-1)d_n}}$ and $\exponent{- n (n
    - 1) d_n }$ are asymptotically equal to $\exponent{ - \frac{1}{
      \left(1 + \frac{1}{n-1} \right)\frac{1}{c_n} - \frac{1}{n}} }$
  and $\exponent{- \left(1 - \frac{1}{n}\right) \frac{1}{c_n} }$
  respectively. Thus $\prob{\min_{2 \leq i \leq n} V_i > d_n } \to 1.$
  For the `only if' part, let $d_n \geq 1 / (c n^2),$ for some real
  constant $c > 0.$ Then for any $n > 1,$ $\prob{\min_{2 \leq i \leq
      n} V_i > d_n} \ \leq \ \exponent{- (1 - \frac{1}{n}) \frac{1}{c}
  } = 1 - \epsilon$ where $\epsilon = \exponent{- (1 - \frac{1}{n})
    \frac{1}{c} } > 0.$ The proof of Lemma~\ref{lemma:rnd:min:dist} is
  complete.
\end{proof}
\begin{proof}[of Theorem~\ref{thm:separability_random}]
  Observe that $T_i$ cannot be separated if both $V_{i} < r(n)$ and
  $V_{i+1} < r(n).$ Defining $W_i:=V_{i}+V_{i+1},$ for $2 \leq i \leq
  n-1,$ we see from
  Section~\ref{sec:system_model}\ignore{(pg.~\pageref{unique:coverage})},
  that it is necessary to have $\min_{2 \leq i \leq n-1} W_i > 2 r(n)$
  for separability. Let us now characterize this minimum and prove
  statement~\ref{full:sep:rn} in the following two steps.
  \begin{enumerate}
  \item We first prove that for some finite constant $c > 0,$ if $r(n)
    = 1/cn^2$ then $\prob{\min W_i \geq 2 r(n)}$ $ < 1.$ Let
    $\Interval := [0,1]$ be divided into $k$ equal sized contiguous
    intervals, referred to as bins in this proof. Recall that the
    target locations are chosen uniformly i.i.d.~in $\Interval.$ The
    event $\left( \min_i \ W_i \ \geq \frac{2}{k} \right)$ implies
    that there exists at most $2$ target locations in any $2$
    consecutive bins, i.e.
    \begin{equation}
      \left( \min_i \ W_i \ \geq \frac{2}{k} \right) \
      \Rightarrow \ \mbox{there are at most two target locations
        in any two consecutive bins}.
      \label{eqn:event:equiv}
    \end{equation}
    This is illustrated in Fig.~\ref{fig:event:wi:geq:2k}. Adjacent
    target locations could be in the same bin (See $T_i, T_{i+1}$ in
    Fig.~\ref{fig:event:wi:geq:2k}) or adjacent bins (See $T_j,
    T_{j+1}$ in Fig.~\ref{fig:event:wi:geq:2k}). The possible
    locations of targets $T_{i-1}, T_{i+2}, T_{j-1}, T_{j+2}$ such
    that $\min W_i \geq 2/k$ are shown as shaded regions in
    Fig.~\ref{fig:event:wi:geq:2k}. The proof
    of~\eqref{eqn:event:equiv} thus follows.
    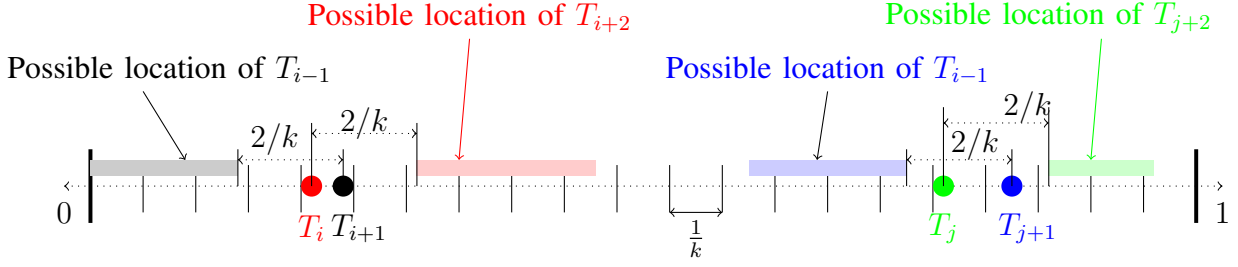
\begin{figure}
      \begin{center}
        \begin{tikzpicture}[scale=0.7]
          \draw [dotted, <->] (-0.5,0) -- (21.5,0); 
          \foreach \x in {0,1,...,21} \draw (\x,0.4) -- (\x, -0.5);
          \node at (-0.5,-0.5) {$0$};
          \node at (21.5,-0.5) {$1$};
          \draw [line width=1.5pt] (0, -.7) -- (0, .7);
          \draw [line width=1.5pt] (21, -.7) -- (21, .7);
          \draw[<->] (11,-0.5) -- (12,-0.5);
          \node at (11.5, -1) {$\frac{1}{k}$};
          
          \fill [color=red] (4.2,0) circle (.2);
          \node [color=red] at (4.2, -.8) {$T_i$};
          \fill [color=red!20] (6.2,0.2) rectangle (9.6,.5);
          \node [color=red] at (7.2, 3.2) {Possible location of $T_{i+2}$};
          \draw [color=red, ->] (7.2, 2.8) -- (7, 0.5);
          \draw [dotted, <->] (4.2,1) -- (6.2, 1);
          \node at (5.2, 1.2) {$2/k$};
          \draw (4.2,0) -- (4.2, 1.2);
          \draw (6.2,0) -- (6.2, 1.2);

          \fill [color=black] (4.8,0) circle (.2);
          \node [color=black] at (5.1, -.8) {$T_{i+1}$};
          \fill [color=black!20] (2.8,0.2) rectangle (0,.5);
          \node [color=black] at (1.5, 2.2) {Possible location of $T_{i-1}$};
          \draw [->] (1,1.8) -- (1.8, 0.5);
          \draw [dotted, <->] (4.8,.5) -- (2.8, .5);
          \node at (3.5, .9) {$2/k$};
          \draw (4.8,0) -- (4.8, .7);
          \draw (2.8,0) -- (2.8, .7);

          \fill [color=green] (16.2,0) circle (.2);
          \node [color=green] at (16.2, -0.8) {$T_{j}$};
          \fill [color=green!20] (18.2,0.2) rectangle (20.2,.5);
          \node [color=green] at (18.2, 3.2) {Possible location of $T_{j+2}$};
          \draw [color=green, ->] (19.2,2.8) -- (19, 0.5);
          \draw [dotted, <->] (18.2,1.2) -- (16.2, 1.2);
          \node at (17.8, 1.5) {$2/k$};
          \draw (18.2,0) -- (18.2, 1.5);
          \draw (16.2,0) -- (16.2, 1.5);

          \fill [color=blue] (17.5,0) circle (.2);
          \node [color=blue] at (17.8, -0.8) {$T_{j+1}$};
          \fill [color=blue!20] (15.5,0.2) rectangle (12.5,.5);
          \node [color=blue] at (14, 2.2) {Possible location of $T_{i-1}$};
          \draw [->] (14,1.8) -- (13.8, 0.5);
          \draw [dotted, <->] (15.5,.5) -- (17.5, .5);
          \node at (16.8, .8) {$2/k$};
          \draw (17.5,0) -- (17.5, .7);
          \draw (15.5,0) -- (15.5, .7);
        \end{tikzpicture}
        \caption{Illustrating the implication
          in~\eqref{eqn:event:equiv} that $\min W_i \geq 2/k
          \Rightarrow$ there are $\leq 2$ target locations in any two
          consecutive bins.}
        \label{fig:event:wi:geq:2k}
      \end{center}
    \end{figure}

    Number the bins starting at $1$ from the left. Let $Y_i$ be the
    indicator variable that collectively in bins $i$ and $i+1$ there
    are at most $2$ target locations.
    \begin{displaymath}
      Y_i = \mathbbm{1}\left(\exists \leq 2 \mbox{ target locations in
          bins } \{i, i+1\} \right).
    \end{displaymath}
    Applying the Chernoff bound to a binomial random variable with
    parameters $(n, 2/k),$ the following bound on $\prob{Y_i = 1}$ is
    obtained.
    \begin{displaymath}
      \prob{Y_i = 1} \leq \left(\frac{n}{k}\right)^2 \left(\frac{1 -
          \frac{2}{k}}{1 - \frac{2}{n}} \right)^{n-2}.
    \end{displaymath}
    Using $k = c n^2$ in the preceding expression, we have
    \begin{equation*}
      \hspace{-2em} \prob{Y_i = 1} \leq \left(\frac{1}{c n}\right)^2
      \left(1 + \frac{2 (cn - 1)}{cn (n-2)} \right)^{n-2} \leq \exponent{2
        \left(1 - \frac{1}{cn}\right) - 2 \log n - 2 \log c} \to 0.
    \end{equation*}
    The second inequality uses $1 + x < e^x.$ Now using the preceding
    relation between the events and then the Markov inequality,
    \begin{displaymath}
      \prob{\min W_i \geq \frac{2}{cn^2}} \ \leq \
      \prob{\sum_{i=1}^{k-1} Y_i \geq k - 1} \ \leq \ \EXP{Y_i}.
    \end{displaymath}
    Combining this result with the Chernoff bound on $\prob{Y_i = 1},$
    we conclude that if $r(n) = 1/cn^2$ for some finite positive
    constant $c,$ then $\prob{\min W_i \geq 2r(n)} \to 0.$
  \item From Lemma~\ref{lemma:rnd:min:dist}, see that if $r(n) = 1/
    \left( n^2 c_n \right),$ for some $c_n \to \infty,$ then
    $\prob{\min V_i \geq r(n)} \to 1$ which implies $\prob{\min W_i
      \geq 2 r(n)} \to 1.$
  \end{enumerate}
  This completes the proof of statement~\ref{full:sep:rn}. We now
  prove statement~\ref{full:sep:mn}. Let $r(n)=1/(c_n n^2),$ divide
  $\Interval$ into $1 / (2r(n)) = c_n n^2/(2 )$ intervals of equal
  width. Every subinterval contains at most one target, with high
  probability. Recall that for such $r(n)$ coverage implies
  identifiability. Then analogous to full separability of
  targets-on-grid model, it is necessary and sufficient to have at
  least one sensor in all the subintervals that contain a target.
  Thus for full separability of uniformly distributed targets, sensors
  with $r(n) = \frac{1}{c_n n^2},$ where $c_n \to \infty,$ $m(n) =
  (c_n n^2/(2)) \ ( \log \ (n^2c_n/(2)) \ + f_n ),$ for any $f_n \to
  \infty,$ is necessary and sufficient. This completes the proof of
  statement~\ref{full:sep:mn}. Remark~\ref{rem:partial:rnd} is the
  prelude to the proof of statement~\ref{partial:sep:rn}.
  \begin{remark}
    \label{rem:partial:rnd}
    In partial separability, since at least $\alpha n$ target
    locations are identifiable, the number of sensors needed is
    clearly not sub-linear. Our strategy thus far has been to divide
    $\Interval$ into contiguous non-overlapping cells such that $r(n)$
    is less than half of cell width and choose $m(n)$ such that there
    is at least one sensor in each cell. Following this process, there
    are two approaches.
    \begin{enumerate}
    \item Choose a small cell size such that all targets are alone in
      their cells and then uniquely cover at least $\alpha n$ of the
      cells containing targets.
    \item Choose a large cell size so that at least $\alpha n$ target
      locations are alone in their cells and choose $m(n)$ to uniquely
      cover all the cells.
    \end{enumerate}
    We adopt the latter approach in the next proof. Recall that
    in the targets-on-grid model, the $m(n)$ required for partial
    separability is lesser than the full separability case by a factor
    of $\log n.$ Thus, it can be expected that the critical number of
    sensors for partial separability of randomly deployed target
    locations will be, in the order sense, smaller than $c_n n^2 \log
    n$ for any $c_n \to \infty.$
  \end{remark}

  We now prove statement~\ref{partial:sep:rn}. Recall the definition
  of spacings from Section~\ref{sec:math}, $V_i = T_i - T_{i-1},$ and
  define $Z_i$ as the indicator variable corresponding to the $i$-th
  target location as follows:
  \begin{equation*}
    Z_i := \indicator \left( V_{i} > r(n) \mbox{ and } V_{i+1} >
      r(n) \right). 
  \end{equation*}
  From~\eqref{eq:kmarginals}, $\EXP{Z_i} = \prob{ Z_i = 1} = (1 - 2
  r(n))^n$. Using~\eqref{eq:reversemarkov}, for any $\alpha_1$ such
  that $\alpha < \alpha_1 < 1,$ we have:
  \begin{equation}
    1 - \frac{1 - \left(1 - 2 r(n) \right)^n}{1 - \alpha_1} \leq
    \prob{\sum_{i=1}^n Z_i \geq \alpha_1 n } \leq \frac{\left( 1
        - 2 r(n) \right)^n}{\alpha_1}. 
    \label{eq:partial:rn:zi}
  \end{equation}
  We first prove statement~\ref{partial:sep:rn:suff}. Let $r(n)$ be
  chosen such that $r(n) < 0.5 /\left( \left(n/c_1\right) + 1\right).$
  Then using the definition of $c_1,$ the following equivalence is
  direct.
  \begin{equation}
    r(n) < \frac{0.5}{ \frac{n}{c_1} + 1} \ \iff \
    \exponent{-\frac{n}{\frac{1}{2r(n)} - 1} } > 1 - (1 - \alpha_1)(1 -
    \beta).
    \label{eq:temp:2}
  \end{equation}
  Thus using~\eqref{eq:partial:rn:zi} and~\eqref{eq:braveinequality}
  in the second inequality of~\eqref{eq:temp:2}, we have
  \begin{eqnarray*}
    \prob{\sum_{i=1}^n Z_i \geq \alpha_1 n } & \geq & 1 - \frac{1 -
      \left(1 - 2 r(n) \right)^n}{1 - \alpha_1} > \beta.
  \end{eqnarray*}
  This completes the proof of statement~\ref{partial:sep:rn:suff}. To
  prove statement~\ref{partial:sep:rn:nec}, let $r(n) > \frac{1}{2n}
  \log \left(\frac{1}{\alpha_1 \beta} \right),$ then
  from~\eqref{eq:partial:rn:zi}, $\prob{\sum_{i=1}^n Z_i \geq \alpha_1
    n } \ \leq \ (1 - 2 r(n))^n/ \alpha_1 \ < \ \beta.$ This completes
  the proof of statement~\ref{partial:sep:rn}.
  
  Before the proof of statement~\ref{partial:sep:mn}, we first obtain
  bounds on the probability of having at least $\alpha n$ identifiable
  target locations as in~\eqref{eqn:wn:markov}. Towards that, let
  $r(n) = c_1 /(4n),$ choose a constant $a > 1,$ and define the
  indicator random variable $W_i$ as follows.
  \begin{equation*}
    W_i := \indicator \left( V_i > a r(n) \ \& \ V_{i+1} > a r(n) \ \&
      \ \exists \mbox{ at least 1 sensor that uniquely senses the
        target} \right).
  \end{equation*}
  Once again for notational convenience, we use $m$ instead of $m(n)$
  for the rest of this proof. Since the target locations and sensor
  locations are chosen independently, we have
  \begin{eqnarray*}
    \prob{W_i = 1} & = & \prob{V_i > a r(n) \ \& \ V_{i+1} > a r(n)}
    \prob{\geq 1 \mbox{ sensors uniquely sense target } i} \\
    & = & \left( 1 - 2 a r(n) \right)^n \ \left( 1 - \left( 1 - 2 (a -
        1) r(n) \right)^m \right).
  \end{eqnarray*}
  The bounds on $\prob{W_i = 1}$ are obtained
  using~\eqref{eq:braveinequality} in the preceding expression.
  \begin{equation}
    \small{ \exponent{- \frac{2 a n r(n)}{ 1 - 2 a r(n)} -
        \frac{1}{e^{2 \left(a - 1 \right) m r(n)} -1} } \leq \prob{W_i = 1}
      \leq \exponent{ - 2 a n r(n) - \exponent{ -\frac{m }{\frac{1}{2
              \left(a - 1 \right) r(n)} - 1}} }}.
    \label{eqn:wn:ineq}
  \end{equation}
  Using~\eqref{eq:reversemarkov} in~\eqref{eqn:wn:ineq}, we have
  \begin{equation}
    \frac{\prob{W_i=1} - \alpha}{1 - \alpha} \leq \prob{\sum_{i=1}^n
      W_i \geq \alpha n} \leq \frac{\sum_{i=1}^n \prob{W_i = 1}}{\alpha n}.
    \label{eqn:wn:markov}
  \end{equation}
  Next, we prove statement~\ref{partial:sep:mn:suff}. Using $\theta_1
  < 2 n r(n)/ c_1 < \theta_2,$ for large $n,$ where $n > 2 a \theta_2
  c_1,$ we have the first implication.
  \begin{eqnarray*}
    & & \hspace{-1in} m \geq \left(\frac{n}{c_1 \theta_1 \left(a-1
        \right)} \right) \log \left(1 + \frac{1}{c_2 - 2 a \theta_2 c_1}
    \right) \\
    & \Rightarrow & m \geq \left(\frac{1}{2 \left(a-1 \right) r(n)}
    \right) \log \left(1 + \frac{1}{\log \left(\frac{1}{1 - \left( 1 -
              \alpha\right)\left(1 - \beta \right)} \right) - \frac{2 n a r(n)}{1 -
          2 a r(n)}} \right). \\
    & \iff & \frac{1}{e^{2 m \left( a -1 \right) r(n)} - 1} \leq \log
    \left( \frac{1}{1 - \left(1 - \alpha\right)\left(1 - \beta\right)
    } \right) - \frac{2 a n r(n)}{1 - 2 a r(n)}. \\
  & \iff & \exponent{- \frac{2 a n r(n)}{1 - 2 ar(n)} -
    \frac{1}{e^{2 \left(a-1\right) m r(n)} - 1} } \geq \alpha +
  (1 - \alpha) \beta.
  \end{eqnarray*}
  The second and third equivalences are obtained by rearranging
  terms. Thus using the final expression and~\eqref{eqn:wn:ineq}
  in~\eqref{eqn:wn:markov}, we have
  \begin{displaymath}
    \prob{\sum_{i=1}^n W_i \geq \alpha n} \geq \frac{\exponent{-
        \frac{2 a n r(n)}{1 - 2 a r(n)} - \frac{1}{\exponent{2 m \left( a -1
            \right) r(n)} - 1} } - \alpha}{1 - \alpha} \geq \beta.
  \end{displaymath}
  This completes the proof of
  statement~\ref{partial:sep:mn:suff}. Next we give the proof of
  statement~\ref{partial:sep:mn:nec}.  For $\theta_1 < 2 n r(n)/c_1 <
  \theta_2,$ we have the following using the definition of $c_3.$
  \begin{eqnarray*}
    & & \hspace{-1in} m < \left( \frac{n}{\theta_2 (a-1)c_1} - 1
    \right) \log \left( \frac{1}{c_3 - a \theta_1 c_1} \right) \\
    & \Rightarrow & m < \left( \frac{1}{2\left(a-1 \right)r(n)} - 1
    \right) \log \left( \frac{1}{\log \left( \frac{1}{\alpha \beta}
        \right) - 2 n a r(n)} \right). \\
    & \iff & \exponent{ - \frac{m}{\frac{1}{2 \left(a-1 \right) r(n)}
        - 1}} > \log \left( \frac{1}{\alpha \beta} \right) - 2 a n
    r(n).
  \end{eqnarray*}
  Further using the preceding expression and~\eqref{eqn:wn:ineq}
  in~\eqref{eqn:wn:markov}
  \begin{displaymath}
    \prob{\sum_{i=1}^n W_i \geq \alpha n} \leq \frac{\exponent{ - 2 a
        n r(n) - \exponent{ - \frac{m}{\frac{1}{2 \left( a -1 \right) r(n)} -
            1}} }}{\alpha} < \beta.
  \end{displaymath}
  This completes the proof of Theorem~\ref{thm:separability_random}.
\end{proof}
\begin{remark}
  \ignore{From Theorem~\ref{thm:separability_deterministic}, for
    targets-on-grid model, the sensing radius $r(n) < 1/n$ for both
    full and partial separability and $m(n) \in \Theta(n \log n)$ for
    full separability and $m(n) \in \Theta(n)$ for partial
    separability. To achieve full separability of random targets we
    need $r(n) = \theta/(c_n n^2),$ and $m(n) = (c_n n^2 / (2 \theta))
    \log (n + f_n),$ for $0 < \theta < 0.5$ and any $c_n, f_n \to
    \infty.$ }
  The following natural schemes for partial separability exist:
  \begin{enumerate}
  \item Following the partial separability of targets-on-grid model,
    choose the sensing radius $r(n)$ such that $\alpha n$ target
    locations are at least $2r(n)$ away from their adjacent neighbors
    with probability $\geq \beta$, and cover those particular $\alpha
    n$ target locations with high probability. This will require $r(n)
    = \theta /(2n)$ and $m(n) = (n/\theta) \log (n/ \theta),$ for some
    constant $\theta > 0.$
  \item Following the full separability of random targets, choose the
    sensing radius $r(n)$ such that all $n$ target locations are at
    least $2r(n)$ away from their adjacent neighbors with high
    probability and cover $\alpha n$ of them with probability $\geq
    \beta$. This will require $r(n) = \theta / (c_n n^2)$ and $m(n) =
    \tilde{\theta} c_n n^2,$ for some constants $\theta,
    \tilde{\theta}$ and some $c_n \to \infty.$
  \end{enumerate}
  Observe from the partial separability results that we have $r(n) \in
  \Theta(1/n)$ and $m(n) \in \Theta(n),$ which are tighter than both
  the above two approaches.
\end{remark}
\begin{remark}
  Modeling a fixed number of target locations may seem impractical, so
  let target locations be realizations of a homogeneous spatial
  Poisson process of intensity $n,$ independent of sensor
  deployment. The results for separability in this Poisson target
  deployment are similar to the (uniform)random target model. The
  proof is a special case of the proof of
  Theorem~\ref{thm:adversary:grid} with $\gamma = 0$ and is omitted.
\end{remark}
\section{Separability in $2$-dimensions}
\label{sec:2d}

In this section, the region of interest is $\Interval^2 := [0,1]^2.$
Each sensor senses all points within $r(n)$ (circle of radius $r(n)$)
from it. Theorem~\ref{thm:2d} summarizes the results of this section.
\begin{theorem}[Separability in $2$-dimensions]
  \label{thm:2d}
  The sensing region is $\Interval^2$ and $m(n)$ denotes the number of
  sensors that are deployed uniformly i.i.d~ in $\Interval^2.$
  \begin{enumerate}
  \item \label{2d:grid} In the targets on grid model, the set of $n$
    target locations, $\TargetGrid^2,$ are the mid points of the cells
    formed when we tessellate $\Interval^2$ into $n$ square cells,
    each of size $\frac{1}{\sqrt{n}} \times \frac{1}{\sqrt{n}}.$ For
    the targets on grid model, the following are true.
    \begin{enumerate}
    \item $0 < \pi r(n)^2 < (\pi/n)$ is necessary for separability.
    \item Let $\pi r(n)^2 = \pi/4n,$ then
      \begin{enumerate}
      \item If $m(n) \geq \frac{4 n}{\pi} \left( \log
          \left(\frac{4n}{\pi}\right) + c_n \right),$ for any $c_n \to
        \infty,$ then \newline $\prob{\mbox{all $n$ target locations
            are separable}} \to 1.$
      \item If $m(n) \leq \frac{4 n}{\pi} \left( \log
          \left(\frac{4n}{\pi}\right) - c_n \right),$ for any $c_n \to
        \infty,$ then \newline $\prob{\mbox{all $n$ target locations
            are separable}} \to 0.$
      \end{enumerate}
    \item Let $\pi r(n)^2 = \pi/4n.$ Given $0 < \alpha < 1,$ and $0 <
      \beta < 1,$ the following are true.
      \begin{enumerate}
      \item If $m(n) \geq (4n/\pi) \ \log \left( \frac{1}{(1 - \alpha)
            (1 - \beta)} \right) $ then $\prob{\mbox{at least } \alpha
          n \mbox{ targets are separable}} >
        \beta.$~\label{fixedpartialsuff2}
      \item If $m(n) < (4n/\pi - 1) \ \log \left( \frac{1}{1 - \alpha
            \beta} \right) $ then $\prob{\mbox{at least } \alpha n
          \mbox{ targets are separable}} <
        \beta.$~\label{fixed:partial:nec:2}
      \end{enumerate}
    \end{enumerate}
  \item In the random target case, the $n$ target locations, denoted
    by $\TargetRandom^2,$ are deployed uniformly i.i.d.~in
    $\Interval^2.$ The following are true.
    \begin{enumerate}
    \item For full separability of $n$ target locations, it is
      necessary that $r(n)^2 = 1/(c_n n)$ for some $c_n \to
      \infty.$~\label{2d:full:sep:rn}
    \item \label{2d:full:sep:mn} Let $\pi r(n)^2 = 1/n c_n,$ for some
      $c_n \to \infty.$ The following are true
      \begin{enumerate}
      \item For any $g_n \to \infty,$ if $m(n) \leq \left(
          \frac{1}{\pi r(n)^2}\right) \ \left(\log \frac{1}{\pi
            r(n)^2} + g_n \right) = \left( n c_n \right) \ \left(\log
          n c_n + g_n \right),$ \newline then $\prob{\mbox{all $n$
            target locations are separable}} \to 1.$
      \item For any $g_n \to \infty,$ if $m(n) \leq \left(
          \frac{1}{\pi r(n)^2}\right) \ \left(\log \frac{1}{\pi
            r(n)^2} - g_n \right) = \left( n c_n \right) \ \left(\log
          n c_n - g_n \right),$ \newline then $\prob{\mbox{all $n$
            target locations are separable}} \to 0.$
      \end{enumerate}
    \item For any $0 < \alpha_1, \beta < 1,$ let $c_1 := \log \left(
        1/ \left( 1 - \left(1 - \alpha_1\right)\left(1 - \beta\right)
        \right) \right)$ and $a>1$ be a finite
      constant. ~\label{2d:partial:sep:rn}
      \begin{enumerate}
      \item If $\pi r(n)^2 \leq \frac{1}{a^2 \left(\frac{n-1}{c_1} + 1
          \right)},$ then \newline $\prob{\mbox{at least } \alpha_1 n
          \mbox{ targets are more than $a r(n)$ away from adjacent
            neighbors}} \geq \beta.$
      \item If $\pi r(n)^2 > \frac{ \log \left(\frac{1}{\alpha_1
              \beta} \right)}{a^2 n},$ then \newline $\prob{\mbox{at
            least } \alpha_1 n \mbox{ targets are more than $a r(n)$
            away from adjacent neighbors}} < \beta.$
      \end{enumerate}
    \item For a given $0 < \alpha, \beta < 1,$ choose an $\alpha_1$
      such that $\alpha < \alpha_1 < 1.$ Let $c_2 := \log$ $\left( 1/
        \left( 1 - \left(1 - \alpha\right)\left(1 - \beta\right)
        \right) \right),$ $c_3 := \log \left( 1/(\alpha \beta)
      \right)$ and $c_1$ is as defined in
      statement~\ref{2d:partial:sep:rn}. Let $\theta_1 \left(c_1/
        \left(a^2 n \right) \right) \leq \pi r(n)^2 \leq \theta_2
      \left(c_1/ \left(a^2 n \right) \right),$ for any $\theta_1,
      \theta_2$ such that $0 < \theta_1 \leq \theta_2 < 1/ \left(1 +
        \left( c_1/ n\right) \right).$ Choose a finite positive
      constant $a$ such that $a^2 > \max \left\{1, c_2/ (2 \theta_1
        c_1) \right\}.$ ~\label{2d:partial:sep:mn}
      \begin{enumerate}
      \item If $m(n) \geq \left( \frac{n}{\theta_1 (a- 1)^2c_1}
        \right) \log\left(1 + \frac{1}{c_2 - a^2 \theta_2 c_1}
        \right),$ then \newline $\prob{\geq \alpha n \mbox{ target
            locations are separable}} \geq \beta.$
      \item If $m(n) < \left( \frac{n}{\theta_2 (a-1)^2c_1} - 1
        \right) \log \left( \frac{1}{c_3 - a^2 \theta_1 c_1} \right),$
        then \newline $\prob{\geq \alpha n \mbox{ target locations are
            separable}} < \beta.$
      \end{enumerate}
    \end{enumerate}
  \end{enumerate}
\end{theorem}
\begin{figure}
  \begin{center}
    \begin{tikzpicture}
      \draw (0,0) rectangle (6,6);
      \fill (0.95,0.95) rectangle (1.05, 1.05);
      \node at (1.1, 0.9) [left] {$x_a$};
      \draw (1,1) circle (1);
      \draw [->] (1,1) -- (1,2);
      \node at (1.7,1.5) [left] {$r_n$};

      \fill[color=red] (5,2) circle (0.1);
      \node[color=red] at (5.4, 2) [below] {$T_c$};
      \fill (4.8,2.6) rectangle (4.9,2.7);
      \node at (4.9, 2.5) [left] {$x_d$};
      \draw (4.85,2.65) circle (1);
      \fill (5.3,1.2) rectangle (5.2,1.1);
      \node at (5.5, .8) [left] {$x_c$};
      \draw (5.25,1.15) circle (1);

      \fill[color=red] (3,5) circle (0.1);
      \node[color=red] at (2.7, 5.15) [above] {$T_b$};
      \draw[color=red] (3.8,4.6) circle (0.1);
      \node[color=red] at (3.9, 4.15) [right] {$T_a$};
      \fill (3.4,5.2) rectangle (3.5,5.3);
      \node at (3.6, 5.2) [right] {$x_b$};
      \draw (3.45,5.25) circle (1);

      \fill[color=red] (1,4) circle (0.1);
      \draw[color=red] (5,5) circle (0.1);
      \fill (0.4,2.2) rectangle (0.5,2.3);
      \fill (3.4,0.2) rectangle (3.5,0.3);

      \node at (-.1,-.1) [below] {$(0,0)$};
      \node at (6.1,6.1) [above] {$(1,1)$};
    \end{tikzpicture}
    \caption{Illustrating Identifiability in $2$ dimensions. The
      description of targets and sensors are identical to those in
      Fig.~\ref{fig:TargetsAndSensingII}}
    \label{fig:2d}
  \end{center}
\end{figure}
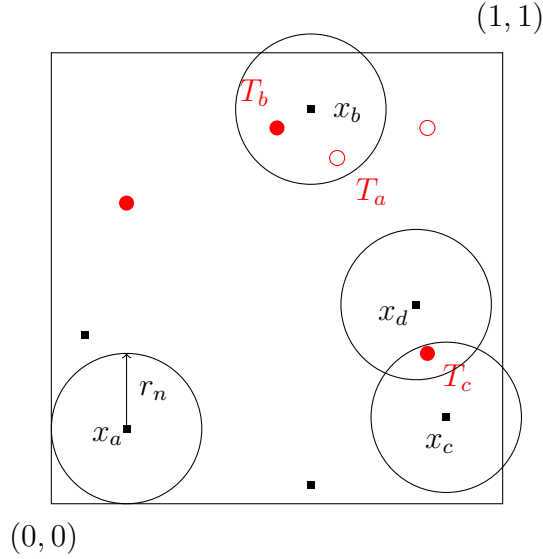
\begin{proof}
  The proof of statement~\ref{2d:grid} is self-evident and is
  omitted. It may be generalized to any $r(n)$ that satisfies the
  necessary condition $0 < \pi r(n)^2 < \pi/n.$

  Next we prove statement~\ref{2d:full:sep:rn}. Divide $[0,1]^2$ into
  $k^2$ equally sized squares, henceforth termed bins. Consider a
  `super-bin' to be a set of $2 \times 2$ adjacent bins
  $\left(\mbox{i.e., of size } \frac{2}{k} \times \frac{2}{k}
  \right).$ Similar to the proof of statement~\ref{full:sep:rn} of
  Theorem~\ref{thm:separability_random}, the following hold.
  \begin{enumerate}
  \item If $k = cn,$ then $\prob{\mbox{each super-bin has } \leq 1
      \mbox{ target location}} \to 0.$ 
  \item If $k = c_n n,$ for some $c_n \to \infty,$ then
    $\prob{\mbox{each super-bin has } \leq 1 \mbox{ target location}}
    \to 1.$ This also ensures that the minimum distance between
    adjacent targets to be of the form $1/nc_n$ for some $c_n \to
    \infty.$
  \end{enumerate}
  Thus we choose $\pi r(n)^2 = 1/ n c_n.$ Another way to see the
  minimum distance condition is as follows. Let $\ErrorEvent_i$ be the
  event that no other target location is within $r(n)$ of the $i$-th
  target location. From the Markov inequality, we know that 
  \begin{displaymath}
    \prob{\sum_{i=1}^n \ErrorEvent_i \geq n} \leq \EXP{\ErrorEvent_i}
    \leq \exponent{- (n-1) \pi r(n)^2 }.
  \end{displaymath}
  If $\pi r(n)^2 = c/n$ for some finite positive constant $c,$ then
  \begin{displaymath}
    \prob{\sum_{i=1}^n \ErrorEvent_i \geq n} \ \leq \exponent{-c
      \left(1 - 1/n \right) }<1.
  \end{displaymath}
  Thus we need $\pi r(n)^2 =
  1/nc_n$ for any $c_n \to \infty$ to ensure that no two target
  locations are within $r(n)$ distance of each other. Thus the proof
  of statement~\ref{2d:full:sep:rn} is complete. The proof of
  statement~\ref{2d:full:sep:mn} is a direct extension of the coupon
  collector result and is omitted.

  Once again for partial separability, the results in
  statements~\ref{2d:partial:sep:rn}, \ref{2d:partial:sep:mn} are
  identical to their one dimensional counterparts. The proofs are
  identically obtained by re-defining $W_i$ as
  \begin{displaymath}
    W_i = \mathbbm{1}\left(\small{\mbox{No other target is within $a
          r(n)$ of target } i \ \& \ \exists \mbox{ at least 1 sensor within
          $(a-1) r(n)$ of target } i }\right).
  \end{displaymath}
\end{proof}
\ignore{For the targets-on-grid model, the conditions on $r(n), m(n)$
  are identical to the one dimensional case because a bijective map
  exists between every interval of width $1/n$ in $\Interval$ and
  every cell of area $1/n$ in $\Interval^2.$}

\section{Separability in the presence of adversarial sensors}
\label{sec:adversaries}

In this section, we consider the sensing area to be $\Interval=[0,1]$
and sensors are deployed according to a spatial Poisson process of
intensity $m$ on $\Interval.$ Let $M (\sim \mbox{Poisson}(m))$ be the
random variable that denotes the number of sensors. In addition, we
assume that a subset, $\Adversaries,$ of the set of sensors
$\Sensors:= \{1,2,\ldots,M\},$ act as adversaries. We also assume that
all sensors report binary observations and the sensor locations are
known apriori. The sensors in the set $\Sensors \backslash
\Adversaries$ report their observations faithfully, we term them
``good'' sensors. The set $\Adversaries$ of adversarial sensors report
an output which may or may not depend on their observation. Each of
the $M$ sensors is an adversary i.i.d.~with probability $\gamma,0 <
\gamma < 1/2,$ independent of anything else; in other words, the good
and adversarial sensors are distributed according to independent
spatial Poisson processes of intensity $(1 - \gamma)m$ and $\gamma m$
respectively. Note that the set of adversaries, $\Adversaries,$ is
unknown to us for the purpose of decoding the target
configuration. First consider the targets-on-grid model, recall the
following results from
Theorem~\ref{thm:separability_deterministic}. Let the sensing radius
be chosen as $r(n) = a/2n$ (or $r(n) = (2-a)/2n)$ for $0 < a \leq 1.$
The following results hold.
\begin{enumerate}
\item $m(n) \geq (n/a) (\log n + c_n)$ guarantees full separability
  for any $c_n \to \infty.$
\item $m(n) \geq (n/a) \log \left(1/ \left[ (1 - \alpha) (1 - \beta)
    \right] \right)$ guarantees partial separability.
\end{enumerate}
\ignore{The output of an adversarial sensor can either be a function
  of only its observation or a function of observations of any subset
  of sensors or arbitrary. }
Observe that without any adversaries, the set of sensors that uniquely
cover a particular target location have the same observations (either
all zero or all one depending on the presence of a target at the
location). Adversaries corrupt the set of sensor observations, and
thus the set of sensors that uniquely sense target $i$ give an
arbitrary binary vector of observations. We assume that the
adversaries don't have knowledge of the number of sensors that
uniquely sense any target. It is easy to see $0 < r(n) < 1/n,$ similar
to statement~\ref{fixed:sep:rn} from
Theorem~\ref{thm:separability_deterministic}, since we need the good
sensor observations to decode the target configuration. We show that
if any $0 < \gamma < (1/2)$ fraction of sensors act as adversaries,
then, with high probability, $m \in \Theta(n \log n)$ ensures that
$\TargetGrid$ is full separable. Consider the following sub-optimal
scheme to decode the target locations from the set of observations and
sensor locations.
\begin{enumerate}
\item To decode the configuration of target $i,$ we only use the
  observations from set of sensors that cover only target $i.$
\item For $0 < \gamma < 1/2,$ we will prove that the number of
  adversarial sensors that uniquely cover target $t_i$ is dominated by
  the number of good sensors that uniquely cover target $t_i, \
  \forall 1 \leq i \leq n.$ Thus `majority decoding' on the set of
  outputs (corresponding to the set of sensors that uniquely cover
  target $t_i$) is necessary and sufficient to decode the state of
  target location $t_i$ for $1 \leq i \leq n$ independent of the
  adversary's behavior.
\end{enumerate}
Let $A_i, G_i,$ respectively be the random variables that denote
number of adversarial sensors and number of good sensors that only
sense target location $i, 1 \leq i \leq n.$ Recall that $A_i, G_i$ are
independent Poisson random variables with intensities $\lambda_A =
\gamma \lambda, \lambda_G = \bar{\gamma} \lambda$ respectively, where
$\bar{\gamma} := 1 - \gamma$ and $\lambda = m/n.$ Let $Q_i := G_i -
A_i.$ The main result of this section is given in
Theorem~\ref{thm:adversary:grid}. We will prove
Theorem~\ref{thm:adversary:grid} for $r(n) = 1/2n,$ however the result
holds for $r(n) = a/2n$ $\left( \mbox{or } (2-a)/(2n)\right)$ for $0 <
a \leq 1,$ with an appropriate change in the constant factors.
\begin{theorem}[Separability of targets-on-grid in the presence of
  adversaries] \label{thm:adversary:grid} Let $r(n) = 1/2n,$ $0 <
  \gamma < 0.5$ and the target locations be $\TargetGrid.$
  \begin{enumerate}
  \item For any $\epsilon > 0,$ if $\frac{m}{n \log n} \ \geq \ \left(
      \frac{1 + \epsilon}{1 - 2 \sqrt{\gamma \bar{\gamma}}} \right),$
    then $\prob{\mbox{all $n$ target locations are separable}} \to
    1.$ \label{adv:full:suff}
  \item For any $0 < \alpha, \beta < 1,$ and any $\epsilon > 0,$ if
    $\frac{m}{n} \geq \left(\frac{1 + \epsilon}{1 - 2 \sqrt{\gamma
          \bar{\gamma}}} \right) \ \log\left( \frac{1}{(1 - \alpha)(1
        - \beta)} \right),$ then \newline $\prob{\mbox{ at least }
      \alpha n \mbox{ target locations are identifiable} } \geq
    \beta.$\label{adv:partial:suff}
  \end{enumerate}
\end{theorem}
\begin{proof}
  We will prove statement~\ref{adv:full:suff} first. Since $A_i$ and
  $G_i$ are independent Poisson random variables, the moment
  generating function (MGF) of $Q_i$ is
  \begin{equation*}
    \mbox{MGF}_{Q_i}(r) = \exponent{- \lambda_G + \lambda_G e^r - \lambda_A +
      \lambda_A e^{-r}} \ \ \forall r \in \Re.
  \end{equation*}
  Since $\mbox{MGF}_{-A}(r) = \mbox{MGF}_A(-r)$ and using the Chernoff
  bound for $-{Q_i} > 0,$ we have
  \begin{displaymath}
    \prob{{Q_i} < 0} \leq \inf_{r \geq 0} \mbox{ MGF}_{Q_i}(-r) =
    \exponent{-\lambda_G - \lambda_A + 2 \sqrt{\lambda_G \lambda_A}}.
  \end{displaymath}
  The complement of this probability gives the required lower bound on
  $\prob{Q_i > 0}.$ Further, use $\lambda_G = \bar{\gamma} \lambda, \
  \lambda_A = \gamma \lambda,$ where $0 < \gamma < 1/2,$ in the
  preceding expression to get
  \begin{equation}
    \prob{{Q_i} > 0} \geq 1 - \exponent{- \left(1 - 2 \sqrt{\gamma
          \bar{\gamma}} \right)\lambda}. \label{eqn:adv:chernoff}
  \end{equation}
  Note that for $0 < \gamma < (1/2),$ the function $\gamma (1 -
  \gamma)$ is concave increasing and has a supremum value of $1/4$ at
  $\gamma = 0.5.$ Thus $1 - 2 \sqrt{\gamma \bar{\gamma}} > 0,$ proving
  that for $\lim_{n \to \infty} \prob{{Q_i} > 0} = \lim_{\lambda \to
    \infty} \prob{{Q_i} > 0} = 1.$ Using independence of $Q_i$ and $c
  := 1 - 2 \sqrt{\gamma \bar{\gamma}},$ we see that
  \begin{equation*}
    \prob{\mbox{all $n$ target locations are separable}} =
    \prod_{i=1}^n \prob{Q_i > 0} \geq \left(1 - e^{- c \lambda} \right)^n.
  \end{equation*}
  Using the inequality in~\eqref{eq:braveinequality}, and $\lambda =
  \frac{m}{n} \geq \frac{(1 + \epsilon) \log n}{c},$ we have:
  \begin{equation*}
    \prod_{i=1}^n \prob{Q_i > 0} \geq \exponent{- \frac{n}{e^{c
          \lambda} - 1}} \geq \exponent{- \frac{1}{n^{\epsilon} - n^{-1}}} \to 1
  \end{equation*}
  thus proving statement~\ref{adv:full:suff}. It is easy to see from
  statement~\ref{fixed:full:mn:nec} of
  Theorem~\ref{thm:separability_deterministic} that if $\frac{m}{n}
  \leq \log n - c_n$ for any $c_n \to \infty,$ then $\prob{\mbox{all
      $n$ target locations are separable}} \to 0.$ Next we prove
  statement~\ref{adv:partial:suff}. Using~\eqref{eq:reversemarkov}
  and~\eqref{eqn:adv:chernoff} with $c = 1 - 2 \sqrt{\gamma
    \bar{\gamma}},$ we see that
  \begin{eqnarray*}
    \prob{\mbox{at least $\alpha n$ target locations are separable}} &
    = & \prob{\sum_{i=1}^n \indicator(Q_i > 0) \geq \alpha n} \\
    & \geq & \frac{ \prob{Q_i > 0} - \alpha}{1 - \alpha} \geq \frac{1
      - \alpha - e^{-c \lambda}}{1 - \alpha}.
  \end{eqnarray*}
  Using $\frac{m}{n} = \lambda > \frac{1}{1 - 2 \sqrt{\gamma
      \bar{\gamma}}} \ \log \left( \frac{1}{(1 - \alpha) (1 - \beta) }
  \right)$ in the preceding expression, it is easy to see that
  \begin{equation*}
    \prob{\mbox{at least $\alpha n$ target locations are separable}}> \beta.
  \end{equation*}
  This completes the proof of statement~\ref{adv:partial:suff}. Using
  the necessary condition from
  Theorem~\ref{thm:separability_deterministic} and arguing as above,
  we see that $m(n) \in \Theta(n).$
\end{proof}
Extending the adversarial setting to the random targets case is
identical to the discrete grid setting discussed above and the results
are similar to their (no adversaries) ideal binary proximity sensor
counterparts, and is hence omitted.

\section{Conclusion}
\label{sec:conclusions}

The separability of an asymptotically large number of static target
locations with binary proximity sensors has been addressed. Target
locations are modeled as a set of deterministic grid points or by
realizations of independent and uniform random variables. Sensor
locations were static and lack of control in their deployment was
modeled by independent and uniform random variables.  Order-optimal
scaling laws for full and partial separability were derived in this
work. For $n$ target locations, where $n \rightarrow \infty$, the
number of sensors needed for full and partial separability in the
deterministic grid case were $\Theta(n \log n)$ and $\Theta(n)$,
respectively.  When target locations are obtained from uniform random
variables, then the number of sensors needed for full and partial
separability were $\omega(n^2 \log n)$ and $\Theta(n)$
respectively. Choices for sensing radius, which is a design parameter,
in various cases were provided. The conditions for separability in two
dimensions were derived. Finally, it was shown that in the presence of
adversarial sensors the scaling laws for separability remain
unaffected.

\bibliography{references}

\begin{thebibliography}{10}

\bibitem{Vaze:Whitespace:12}
Rahul Vaze and Chandra~R. Murthy,
\newblock ``On whitespace identification using randomly deployed sensors,''
\newblock Preprint, 2012.

\bibitem{royleNKGA2009}
A.~Royle, J.~Nichols, U.~Karanth, and A.~Gopalaswamy,
\newblock ``A hierarchical model for estimating density in camera-trap
  studies,''
\newblock {\em Journal of Applied Ecology}, vol. 46, no. 1, pp. 118--127, 2009.

\bibitem{mudumbaiMI2008}
R.~Mudumbai and U.~Madhow,
\newblock ``Information theoretic bounds for sensor network localization,''
\newblock in {\em Proceedings of International symposium on Information
  Theory(ISIT)}, 2008, pp. 1602--1606.

\bibitem{patwariAKHMCL2005}
N.~Patwari, J.~N. Ash, S.~Kyperountas, A.~O. Hero-III, R.~L. Moses, and N.~S.
  Correal,
\newblock ``Locating the nodes: cooperative localization in wireless sensor
  networks,''
\newblock {\em IEEE Signal Processing Magazine}, vol. 22, no. 4, pp. 54--69,
  July 2005.

\bibitem{Simic01}
S.~N. Simic and S.~Sastry,
\newblock ``Distributed localization in wireless ad hoc networks,''
\newblock Tech. {R}ep. Tech. Rep. UCB/ERL M02/26,, UC Berkeley, December 2001.

\bibitem{Karnik04}
A.~Karnik and A.~Kumar,
\newblock ``Iterative localisation in wireless ad hoc sensor
  networks:one-dimensional case,''
\newblock in {\em Proceedings of International Conference on Signal Processing
  and Communications (SPCOM)}, 2004, pp. 209--213.

\bibitem{singhMKSCT2007}
J.~Singh, U.~Madhow, R.~Kumar, S.~Suri, and R.~Cagley,
\newblock ``Tracking multiple targets using binary proximity sensors,''
\newblock in {\em Proceedings of the 6th International Conference on
  Information Processing in Sensor Networks}, 2007, pp. 529--538.

\bibitem{shrivastavaMMST2009}
N.~Shrivastava, R.~Mudumbai, U.~Madhow, and S.~Suri,
\newblock ``Target tracking with binary proximity sensors,''
\newblock {\em ACM Transactions on Sensor Networks}, vol. 5, no. 4, pp.
  30:1--30:33, Nov. 2009.

\bibitem{gandhiKST2008}
S.~Gandhi, R.~Kumar, and S.~Suri,
\newblock ``Target counting under minimal sensing: Complexity and
  approximations,''
\newblock in {\em Algorithmic Aspects of Wireless Sensor Networks: Fourth Intl.
  Workshop (ALGOSENSORS)}, 2008, pp. 30--42.

\bibitem{Hall88}
P.~Hall,
\newblock {\em Introduction to the Theory of Coverage Processes},
\newblock John Wiley, 1988.

\bibitem{Kumar04}
S.~Kumar, T.~H. Lai, and J.~Balogh,
\newblock ``On $k$-coverage in a mostly sleeping sensor network,''
\newblock in {\em Proceedings of 10th ACM MobiCom}, 2004, pp. 144--158.

\bibitem{Liu04}
B.~Liu and D.~Towsley,
\newblock ``A study of the coverage of large-scale sensor networks,''
\newblock in {\em IEEE International Conference on Mobile Ad-hoc and Sensor
  Systems}, 2004, pp. 475--483.

\bibitem{davidno2003}
H.~A. David and H.~N. Nagaraja,
\newblock {\em Order Statistics},
\newblock John Wiley \& Sons, New York, NY, 3nd edition, 2003.

\bibitem{motwaniRR1995}
R.~Motwani and P.~Raghavan,
\newblock {\em Randomized Algorithms},
\newblock Cambridge, New York, NY, USA, 1995.

\end{thebibliography}

\end{document}